\documentclass{article}

\usepackage[english]{babel}

\usepackage{amsthm}
\usepackage{amssymb}
\usepackage{bbm}
\usepackage{caption}
\usepackage{indentfirst}

\usepackage{algorithm}
\usepackage{algpseudocode}

\newenvironment{varalgorithm}[1]
  {\algorithm\renewcommand{\thealgorithm}{#1}}
  {\endalgorithm}

\newtheorem{theorem}{Theorem}[section]
\usepackage{enumitem}
\newtheorem{lemma}[theorem]{Lemma}

\newcommand{\thmtext}{ There exists a sequence of matrices $M_n\in\{0,1\}^{m\times n}$ 
that satisfies:  
    \begin{enumerate}
        \item There exists a valid set of choices for \ref{alg:ess-finder} such that $\ref{alg:ess-finder}(M_n)$ outputs $A_n$.
        \item
        $dim(Ker(M_n)) \leq n/2+2$.
        \item $dim(Ker(A_n))\geq n/2+\Omega(n)$.
    \end{enumerate}}
\newtheorem*{T1}{Theorem~\ref{thm: matrix properties}}

\newtheorem{claim}[theorem]{Claim}
\newtheorem{corollary}[theorem]{Corollary}

\theoremstyle{definition}
\newtheorem{definition}{Definition}[section]

\usepackage[letterpaper,top=2cm,bottom=2cm,left=1.5cm,right=1.5cm,marginparwidth=1.75cm]{geometry}

\usepackage{amsmath}
\usepackage{graphicx}
\usepackage[colorlinks=true, allcolors=blue]{hyperref}
\usepackage[colorinlistoftodos]{todonotes}
\usepackage[utf8]{inputenc}
\newcommand{\set}[1]{\{ {#1} \}}
\newcommand\sett[2]{\left\{ #1 \left| \; \vphantom{#1 #2} \right. #2  \right\}}

\newcommand{\one}{\mathbbm{1}}

\newcommand{\bits}{\{0,1\}}
\newcommand{\remove}[1]{}
\newcommand{\calC}{\mathcal{C}}

\title{The linear time encoding scheme fails to encode}
\author{Yotam Dikstein\thanks{Institute for Advanced Study, USA. email: yotam.dikstein@gmail.com.}, \; Irit Dinur\thanks{Weizmann Institute of Science, ISRAEL. email: irit.dinur@weizmann.ac.il. Supported by ERC grant 772839, and ISF grant 2073/21.}
\;and Shiri Sivan\thanks{Weizmann Institute of Science, ISRAEL. email: shirisiv@gmail.com. }}
\begin{document}
\maketitle

\begin{abstract}
We point out an error in the paper ``Linear Time Encoding of LDPC Codes'' (by Jin Lu and José M. F. Moura, IEEE Trans). The paper claims to present a linear time encoding algorithm for every LDPC code. We present a family of counterexamples, and point out where the analysis fails. The algorithm in the aforementioned paper fails to encode our counterexample, let alone in linear time.
\end{abstract}

\section{Introduction}
A Low Density Parity Check (LDPC) code is defined as the null-space of a low density $m{\times}n$ ($m<n$) matrix over $\mathbb{F}_2$. In this context, a matrix is called low density if each row has $O(1)$ ones. Gallager was the first to study random ensembles of such codes \cite{Gal} and proved that they can be decoded in linear time by a simple message-passing algorithm. Sipser and Spielman \cite{SipSp} showed that this works whenever the parity-check graph is a good enough expander. 

Although decoding is optimal (linear time), the straightforward encoding procedure is quadratic, because the generator matrix of these codes is dense. Are there more efficient algorithms? There are several specific families of codes for which the  encoding complexity has been analyzed. For example:
\begin{itemize}
    \item Spielman \cite{Spielman} constructs a family of linear time encodable and decodable codes. 
    \item Richardson et al. \cite{RU} present a number of encoding schemes (distinguished by their preprocessing algorithms) and analyze their expected performance on various LDPC distributions. They find that certain distributions can be encoded in expected linear time using these algorithms. However, Di et al. \cite{DRU} show that these distributions have expected sub-linear distance.
\end{itemize}
    
    As previously mentioned, Richardson et al. \cite{RU} introduce multiple preprocessing  algorithms. These  algorithms, notably Algorithms C and D, exhibit similarities to those proposed by Lu et al. \cite{LuMoura2008}. Both papers aim to triangularize the input LDPC matrix, or bring it as close as possible to triangular form, through greedy row and column permutations. The primary distinction between the algorithms suggested by \cite{LuMoura2008} and \cite{RU} lies in the fact that the former allows a restricted number of row additions along with row and column permutations, while the latter does not. \cite{RU} calculate an expected \textbf{quadratic} encoding complexity for  certain LDPC distributions, whereas  \cite{LuMoura2008} claim that their algorithm assures linear time encoding for all LDPC codes. 

Excluding Lu and Moura \cite{LuMoura2008}, there have been no claims of a sub-quadratic encoding algorithm for general LDPC codes. 
In this note we address the algorithm presented in \cite{LuMoura2008} and show that it contains a critical flaw. Specifically, their algorithm is indeed linear-time, but fails to encode the code given by the input LDPC matrix.

\begin{theorem}
    There exists a family of LDPC codes given by  matrices $M_n$ on which the algorithm  presented in \cite{LuMoura2008} fails to encode the code.
\end{theorem}
In particular, the algorithm of \cite{LuMoura2008} fails on the $9\times 18$ matrix $M_{18}$ depicted in the following figure, as we illustrate in Section \ref{sec:running-the-alg}.
\begin{figure}[htp]
    \centering
    \includegraphics[scale=0.3]{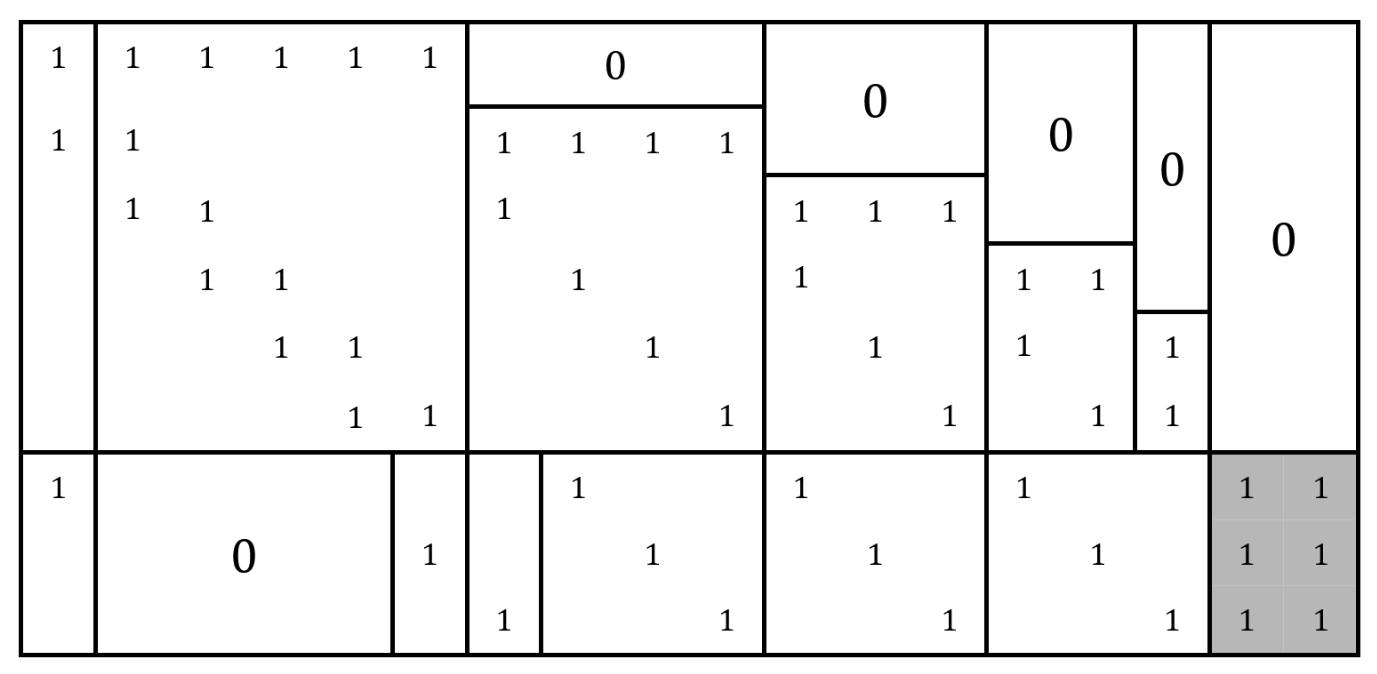}
    \caption{$M_{18}$. All blank spaces are zeros.}
    \label{fig:M18}
\end{figure}

\section{Definitions}
\begin{definition}[Kronecker product]\label{def:otimes}
    If $A$ is an $m\times n$ matrix and $B$ is a $p \times q$ matrix, then the Kronecker product $A \otimes B$ is the $pm \times qn$ block matrix:
    \begin{figure}[htp]
    \centering
    \includegraphics[width=5cm, 
    ]{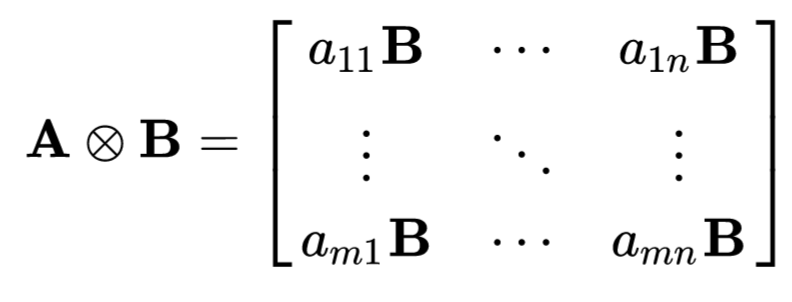}
\end{figure}
\end{definition}

A linear code $\calC$ specified by a {\em parity check} matrix $M\in \{0,1\}^{m\times n}$, is the linear subspace 
\[ \calC = \sett{x\in \bits^n}{Mx=0}.\]
This space is also referred to as the kernel of $M$, denoted as $Ker(M)$.
We like to view $M$ also as the adjacency matrix of a bipartite graph 
with left nodes $[n]$ called variables and right nodes $[m]$ called constraints such that $M(i,j)=1$ iff the $i$-th variable is connected to the $j$th constraint.

Lu and Moura mainly use the graph representation, while we prefer the matrix view. In Section \ref{sec: algs} we will present the algorithms from \cite{LuMoura2008} in both languages. We will freely alternate between the terms variable and column, constraint and row.

Given a matrix $M\in\bits^{m\times n}$, a set of row indices $C\subset[m]$ and column indices $V\subset[n]$ we use the notation $M(C,V)$ to denote the sub-matrix of $M$ induced by these sets. This corresponds to the subgraph of the Tanner graph induced by vertices $V\subset [n]$ and $C\subset [m]$.

We will use the following definition of an (algebraic) circuit over \(\mathbb{F}_2\). 

\begin{definition}[Circuit]
    A circuit \(\Phi = (V,E)\) is a directed acyclic graph such that every vertex has in-degree at most \(2\) (i.e. fan-in \(2\)). The input of the circuit, \(I(\Phi) \subseteq V\) are all vertices with in-degree \(0\). The output of the circuit \(O(\Phi)\) are all vertices with out-degree \(0\). The size of the circuit \(|\Phi| = |V|\).
\end{definition}
Vertices of a circuit are sometimes called gates. We note that while formally, the fan-in in this model is \(2\). The results in this paper remain the same by relaxing the fan-in bound to any constant \(d\) instead.

Let \(\Phi\) be circuit with input vertices \(I=\set{v_1,v_2,\dots,v_k}\) and output vertices \(O = \set{u_1,u_2,\dots,u_m}\). A circuit naturally calculates a linear function \(T:\mathbb{F}_2^k \to \mathbb{F}_2^m\) as follows. Given \((x_1,x_2,\dots,x_n) \in \mathbb{F}_2^n\) every input vertex \(v_i\) is labeled with \(x_i\). Then the label of every other vertex is set to be the sum of the labels of its incoming neighbors (mod \(2\)). The fact that \(\Phi\) is acyclic ensures that such a labeling is possible. Finally, the value of \(T(x_1,x_2,\dots,x_k) = (\ell(u_1),\ell(u_2),\dots,\ell(u_m))\) where \(\ell(u_j)\) is the labeling of \(u_j\).

\begin{definition}[Linear time encodable codes]
    An infinite family of matrices \(\set{M_n \in \bits^{m_n\times n}}_n\) is \emph{linear time encodable} if there exists a constant \(c > 0\) and circuits \(\set{\Phi_n}_n\) of size at most \(c n\) such that $\Phi_n$ calculates a linear isomorphism \(T_n:\mathbb{F}_2^{k_n} \to Ker(M_n)\).
\end{definition}

Lu and Moura \cite{LuMoura2008} first present a construction of linear sized circuits and characterize the codes that are encodable by this construction. These codes are those whose code graph does not contain certain subgraphs called Encoding Stopping Sets (ESS) or Pseudo Encoding Stopping Sets (PESS). Graphs without such subgraphs are called Pseudo-Trees. We beleive this part of their paper is correct, see \cite[Corollary 1]{LuMoura2008} which refers to connected graphs but this easily generalizes to unions of such.

Then \cite{LuMoura2008} introduces an algorithm that takes as input an LDPC matrix \(M\) and outputs a linear sized circuit \(\Phi\) that is supposed encode \(M\). They do so by decomposing the matrix \(M\) into submatrices that are encodable via their initial construction (or a slight modification of it). We will show that this decomposition fails. Doing so requires a few more definitions.
 
 \begin{definition}[(Pseudo) Encoding Stopping Set]
Let \(V \subseteq [n]\) and \(C \subseteq [m]\) be subsets of the columns and rows of a matrix $M$. An Encoding Stopping Set (ESS) is a submatrix $M(C,V)$ such that:
\begin{enumerate}
    \item For all $c \in C$,  $\sett{v}{M(c,v) = 1} \subseteq V$. That is, all variables participating in this constraint are in \(V\).
    \item For all $v\in V$ the Hamming weight of the \(v\)-th column is at least \(2\) in \(M(C,V)\). That is, every variable in \(v\in V\) participates in at least two constraints in \(C\).
    \item The set of rows corresponding to $C$ is linearly independent (in other words, $M(C,[n])$ has rank $|C|$ over $\mathbb{F}_2$).
    
\end{enumerate}
If items \(1,2\) hold but item 3 does not, we call the submatrix a Pseudo Encoding Stopping Set (or PESS).
\end{definition}

\begin{definition}[Pseudo-Tree]
\label{def: pseudo-tree}
    A matrix is a Pseudo-Tree if it does not contain an ESS or a PESS.
\end{definition}
Lu and Moura also require that the Pseudo-Tree is connected (as a graph), but this requirement is not necessary for our purposes. Additionally, they define Pseudo-Trees in different terms, but they show the equivalence to our definition (see \cite[Corollary 1]{LuMoura2008}).

As mentioned, Pseudo-Trees are linear-time encodable via a greedy algorithm, see \cite[Lemma 1]{LuMoura2008}. Lu and Moura observe that if an ESS is ``almost" a Pseudo-Tree, then it is linear-time encodable. By ``almost" we mean that there is a constant number of constraints whose removal yields a Pseudo-Tree. Thus, they also give the following definition.
\begin{definition}[k-fold-constraint (Pseudo) Encoding Stopping Set]
Let $M(C,V)$ be an (P)ESS. We say that $M(C,V)$ is a k-fold-constraint Encoding Stopping Set if the following two conditions hold. 
\begin{enumerate}
    \item There exists k constraints $c_1,...,c_k \in C$ s.t. $M(C\setminus\set{c_1,...,c_{k}}, V)$ does not contain any PESS or ESS.
    \item For any $k-1$ constraints $c_1,...,c_{k-1}$, $M(C\setminus\set{c_1,...,c_{k-1}}, V)$ contains a PESS or ESS.
\end{enumerate}
\end{definition}

Lu and Moura provide a linear time encoding algorithm for any \(1\)- or \(2\)-fold constraint (P)ESS \cite[Algorithm 4]{LuMoura2008}.

\subsection{Lu and Moura's algorithms}
\label{sec: algs}
We now present the two main algorithms used in Lu and Moura's paper. We will present a graph version as well as a matrix version for these algorithms.
The Algorithm \ref{alg: ess finder- graph} finds a PESS or a 1 or 2-fold-constraint ESS in a given bipartite graph (see \cite[Algorithm 5]{LuMoura2008} for Lu and Moura's version). The second algorithm, \ref{alg:decomposition-graphs} utilizes the algorithm \ref{alg: ess finder- graph} to decompose the given graph into linear-time encodable components (see \cite[Algorithm 6]{LuMoura2008}).

We give a more streamlined description of their algorithms. In particular, we added the algorithm \ref{alg: strip} as a sub-procedure of \ref{alg:ess-finder} for easier reference, and we omit the analysis from the description of the algorithms. 

\begin{varalgorithm}{(P)ESS-FINDER}
\caption{\cite[Algorithm 5]{LuMoura2008}, 
Find a 1 or 2-fold-constraint (P)ESS- graph language}
\label{alg: ess finder- graph}
Input: A Tanner graph \(G=(V\cup C,E)\) with maximal variable degree 3.
\begin{enumerate}
    \item Initialize \(H=(V_H\cup C_H,E_H) \leftarrow G\), \(S \leftarrow \emptyset\).
    \item While \(H \ne \emptyset\):
    \begin{enumerate}
        \item Choose a lightest constraint \(c \in C_H\). Add \(c\) and its neighbours to \(S\) and remove them from \(H\).
         \item If \(c\) doesn't have any neighbours in \(H\) and \ref{alg: strip}\((S) \ne \emptyset\) then output \ref{alg: strip}\((S)\).

    \end{enumerate}
    \item Return \(S\). 
\end{enumerate}
\end{varalgorithm}

\begin{varalgorithm}{STRIP}
\caption{Remove degree one variables}
Input: Graph \(S\).
\label{alg: strip}
\begin{enumerate}
    \item While there exists a degree \(1\) bit node \(x \in V_S\):
    \begin{enumerate}
        \item Remove \(x\) and its neighbours from \(S\).
    \end{enumerate}
    \item Return \(S\).
\end{enumerate}
\end{varalgorithm}

\begin{varalgorithm}{(P)ESS-FINDER}
\caption{\cite[Algorithm 5]{LuMoura2008}, 
Find a 1 or 2-fold-constraint  (P)ESS- matrix language}
\label{alg:ess-finder}
\textbf{Input:}
An LDPC matrix $M \in \set{0,1}^{m\times n}$ with maximal column weight 3.
\begin{enumerate}
    \item Initialize $H\leftarrow{M}, C\leftarrow\emptyset, V\leftarrow\emptyset$.
    \item While $C \ne [m]$ and $V \ne [n]$:
    \begin{enumerate}
        \item Choose a lightest row $c$. Let $V_c$ be the indices of the corresponding variables. Add $c$ to $C$ and $V_c$ to $V$ and zero the columns $V_c$ in $H$. 
        \item \label{step: no new vars}
        If $V_c=\emptyset$, and $STRIP(M(C,V))\neq\emptyset$ then output $STRIP(M(C,V))$.
    \end{enumerate}   
    \item Return $M(C,V)$.
\end{enumerate}
\end{varalgorithm}

The $STRIP$ procedure used in the matricial version of \ref{alg:ess-finder} is equivalent to the one defined for graphs, namely we remove from $V$ any column of weight one (and the corresponding row from $C$) and repeat.

\begin{varalgorithm}{DECOMPOSE}\caption{\cite[Algorithm 6]{LuMoura2008}, Decompose a PC to ESS's, PESS's and encodable components} \label{alg:decomposition}
Input: An LDPC matrix $M(C, V)$.
\begin{enumerate}
    \item Initialize $i\leftarrow 1$, $Components\leftarrow \emptyset$, $M(C_1, V_1)\leftarrow \;\ref{alg: ess finder- graph}(M(C,V))$.
    \item While $M(C_i, V_i)\neq M(C,V)$:
    \begin{enumerate}
        \item If \(M(C_i, V_i)\) is a PESS: 
        \begin{enumerate}
            \item Find \(C' \subseteq C_i\), such that \(\sum_{c \in C'} c|_{V_i} = 0 \text{ (modulo 2)}\).
            \item\label{step:removal} Choose a constraint $c\in C'$ and remove it from $C_i$.
            \item If $i>1$:
            \begin{enumerate}
                \item Remove $M(C_{i-1}, V_{i-1})$ from $Components$.
                \item \label{step: add constraint}Add the constraint $c^*=\sum_{c\in C'}c\big|_{V_{i-1}}$ to  $C_{i-1}$.
                \item Add $ \ref{alg:decomposition}(M(C_{i-1}, V_{i-1}))$ to $Components$.
            \end{enumerate}
        \end{enumerate}
    \item Add $M(C_i, V_i)$ to $Components$.
    \item $C \leftarrow C\setminus C_i$.
    \item $V \leftarrow V\setminus V_i$.
    \item $i\leftarrow i+1$.
    \item $M(C_i,V_i)\leftarrow\ref{alg:ess-finder}(M(C,V))$.  
    \end{enumerate}
\item Add $M(C,V)$ to $Components$.
\item Output $Components$.
\end{enumerate}
\end{varalgorithm}

For the graph version of Algorithm \ref{alg:decomposition-graphs} we use the notation $G(\Tilde{U})$, to denote the subgraph of $G=(U,E)$ induced by $\Tilde{U}\subseteq U$.
\begin{varalgorithm}{DECOMPOSE}
\caption{\cite[Algorithm 6]{LuMoura2008}, Decompose a graph to ESS's, PESS's and encodable components}
\label{alg:decomposition-graphs}

Input: A Tanner graph $G = (V \cup C,E)$.
\begin{enumerate}
    \item Initialize $i\leftarrow 1$, $Components\leftarrow \emptyset$, $G_1 = G(V_1 \cup C_1) \leftarrow Algorithm \;\ref{alg: ess finder- graph}(G)$.
    \item While $G_i \neq G(V\cup C)$:
    \begin{enumerate}
        \item If \(G_i\) is a PESS: 
        \begin{enumerate}
            \item Find \(C' \subseteq C_i\), such that \(\sum_{c \in C'} c|_{V_i} = 0 \text{ (modulo 2)}\).
            \item Choose a constraint variable $c\in C'$ and remove it from $C_i$.
            \item If $i>1$:
            \begin{enumerate}
                \item Remove $G_{i-1}$ from $Components$.
                \item Add the vertex \(c^* = \sum_{c\in C'}c\big|_{V_{i-1}}\) to \(C_{i-1}\).
                \item Add $ \ref{alg:decomposition-graphs}(G_{i-1})$ to $Components$.
            \end{enumerate}
        \end{enumerate}
    \item Add $G_i$ to $Components$.
    \item $C \leftarrow C\setminus C_i$.
    \item $V \leftarrow V\setminus V_i$.
    \item $i\leftarrow i+1$.
    \item $G_i\leftarrow\ref{alg: ess finder- graph}(G(V\cup C))$ 
    \end{enumerate}
\item Add $G$ to $Components$.
\item Output $Components$.
\end{enumerate}
\end{varalgorithm}

\newpage\section{The flaw in the paper}
Before pointing to the error in \cite{LuMoura2008} we provide an outline of their intended encoding strategy. 
\begin{itemize}
    \item Input:
An LDPC matrix $M$ with kernel dimension $k$. It is assumed that each column in $M$ has weight at most $3$, by adding variables, if needed.
\item {Output:}
A linear-sized circuit $\Phi$ that implements $w\mapsto Gw$ for all $w\in \bits^k$, for some matrix $G$ such that $Im G = Ker M$. 
\end{itemize}

The paper's approach for constructing $\Phi$ from $M$ is based on a decomposition algorithm, \ref{alg:decomposition} (see \cite[Algorithm 6]{LuMoura2008}) which generates a list of ``components'' which are pseudo-trees and 1/2-fold-constraint Encoding Stopping Sets using \ref{alg:decomposition}. 

As mentioned earlier, \cite{LuMoura2008} observe that each component, being a pseudo-tree or a  1/2-fold-constraint ESS, admits linear time encoding via the label-and-decide or label-and-decide-recompute algorithm in \cite[Algorithms 2,3,4]{LuMoura2008}. This algorithm shows how to partition the bits into message bits and output bits so that one can propagate the values from message bits to output bits, using the constraints, in linear time. 

The decompose algorithm outputs components together with an implicit labeling of their input and output bits. Every component corresponds to a matrix in the output of Algorithm \ref{alg:decomposition}. However, these can be described as a collection of circuits and connections between them, as portrayed in Figure \ref{figure:phi}. More precisely, the components can be described as a collection of circuits \(\Phi_1,\Phi_2,\dots,\Phi_i\) that can be composed into a circuit \(\Phi\) that encodes the code, such that in this decomposition some of the input bits of every \(\Phi_i\) are connected to some of the output bits of \(\set{\Phi_j : 1 \leq j < i}\). Every circuit \(\Phi_i\) is supposed to encode a pseudo-tree or a \(1\)- or \(2\)-fold (P)ESS, thus the matrices that are outputted in Algorithm \ref{alg:decomposition} should be pseudo-trees or a \(1\)- or \(2\)-fold (P)ESSs.

\begin{figure}[htp]
    \centering
    \includegraphics[scale=0.42]{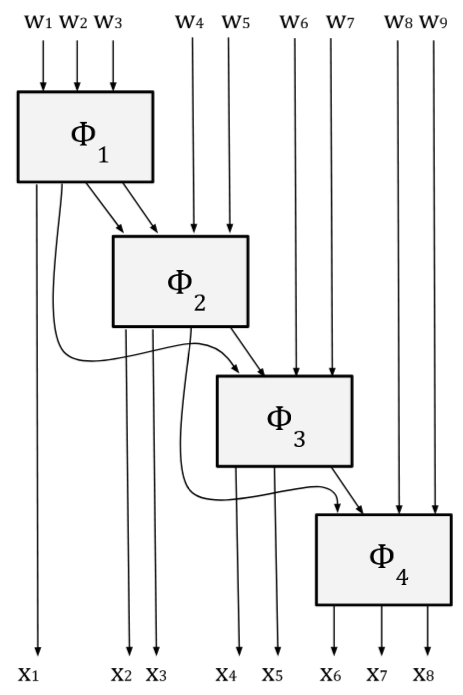}
    \caption{The circuit $\Phi$.}
    \label{figure:phi}
\end{figure}

\subsection*{The Flaw} 
Unfortunately, it is not true that \ref{alg:decomposition} returns \(1\)- or \(2\)-fold-constraint (P)ESS's and pseudo-Trees. In Step \ref{step:removal}, parallel to the third line of the decomposition algorithm in \cite[Algorithm 6]{LuMoura2008}, the authors claim that if \ref{alg:ess-finder} outputs a PESS, then removing one constraint transforms it to a Pseudo-Tree. In other words, \cite{LuMoura2008} falsely assume that the only linear dependency is the sum of all constraints and thus the removal of any single constraint resolves the linear dependency. However, there could be multiple linear dependencies on this same set of variables. The algorithm fails to take these into consideration, thereby resulting in a code that has too many codewords.

The main issue is that the constraints that we failed to add at step \ref{step: add constraint} are never taken into consideration, thereby resulting in a code that has too many codewords.

In the next section we provide an example where  \ref{alg:decomposition} encounters a PESS which has a linear number of constraints that are ignored. 

The first component output by the algorithm thus has a Kernel that is much larger and contains many non-codewords. 

\section{A Counter Example}

\begin{theorem}\label{thm: matrix properties}
There exists a sequence of matrices $M_n\in\{0,1\}^{m\times n}$ with a sequence of sub matrices $A_n$ that satisfy:  
    \begin{enumerate}
        \item \label{subthm: valid choice} There exists a valid set of choices for \ref{alg:ess-finder} such that $\ref{alg:ess-finder}(M_n)$ outputs $A_n$. 
        \item \label{subthm: M_n Nullity}
        $dim(Ker(M_n)) \leq n/2+2$.
        \item \label{subthm: A_n Nullity} $dim(Ker(A_n))\geq n/2+\Omega(n)$.
    \end{enumerate}
\end{theorem}

Let us assume that the encoding scheme suggested by \cite{LuMoura2008} works. This implies that the message bits of the code are a union of the message bits of the components output by \ref{alg:decomposition}. That is, denoting by $k_i$ the number of new input bits that enter the $i$-th component,
\begin{equation}\label{eq:falseclaim}
    dim(Ker(M))=\sum_{i\in[r]}k_i
\end{equation}

As the following corollary shows, this is not always true.

\begin{corollary} \label{cor:decomp-fails}
    There exists a sequence of matrices $M_n\in\{0,1\}^{m\times n}$ and a valid set of choices for $\ref{alg:decomposition}(M_n)$ s.t. $\sum_{i\in [r]} k_i > dim( Ker (M_n))$.
\end{corollary}
\begin{proof}[Proof of Corollary \ref{cor:decomp-fails}, assuming Theorem \ref{thm: matrix properties}]
    Assume towards contradiction that \ref{alg:decomposition} is correct on any input matrix $M$. 

    We assume that \ref{alg:decomposition} and \ref{alg:ess-finder} goes through the constraints in order.
    By inspection, we can see that the first component output is $A_n$ which is an ESS. After removing it from $M_n$, the algorithm continues to decompose the remainder graph. 
    In the next step, due to the structure of $M_n$, the algorithm will find a PESS, and therefore add a a constraint to $A_n$ resulting in $A'_n$, and run \ref{alg:decomposition} on $A'_n$. 

    The output of this step is a list of components  $M'_1,\ldots,M'_t$, such that, assuming that \ref{alg:decomposition} is correct and \eqref{eq:falseclaim} holds, $\sum_{j=1}^t k_j = \dim (Ker A'_n) \geq \dim(Ker(A_{n}))-1$. 

    However, by \autoref{subthm: M_n Nullity} of \autoref{thm: matrix properties}
    $n/2+2\geq dim(Ker(M_n))$, while \autoref{subthm: A_n Nullity} of \autoref{thm: matrix properties} assures that $dim(Ker(A_n))\geq n/2+\Omega(n)$. Combining the above and invoking \eqref{eq:falseclaim} once more, \[n/2+2 \geq dim(Ker(M_n))\geq dim(Ker(A'_n))\geq dim(Ker(A_n))-1 \geq n/2+\Omega(n),\] 
    and this leads to a contradiction to the correctness of \ref{alg:decomposition}.
    
\end{proof}

The proof of \autoref{thm: matrix properties} will follow from the description of the counterexample, along with some necessary properties. Our counter-example has column weight 3 and row weight 6. There are $n$ columns and $m$ rows ($m=n/2$), where $n=11N+7$ for any odd integer $N$. The counter-example $M_n\in \{0,1\}^{m\times n}$ is 
$$
M_n=
\left[
\begin{array}{c|c}
A_n&0\\
\hline
B_n&I_{\frac{N+1}{2}} \otimes \mathbf{1}_{3\times 2}\\
\end{array}
\right]
$$ 
where $\otimes$ stands for the Kronecker product, see Definition \ref{def:otimes}, and we now detail the matrices $A_n,B_n$.
\paragraph{The matrix \(A_n\).} We will now describe $A_n \in \{0,1\}^{(4N+2)\times(10N+6)}$. Towards this we shall decompose \(A_n = S_n + D_n\) (\(S\) stands for ``stairs'' and \(D\) for ``diagonals'') and describe these two components separately. \(S_n\) is the matrix described by blocks in Figure 
\ref{figure:S_n}, where:
\begin{enumerate}
    \item $\forall d\in[4]$, \(T_d = I_{N} \otimes \one_d\) where \(\one_{d} = (1,1...,1)\) \(d\)-times. 
    \item The first row of \(S_n\) is all zero except for six 1's at columns 1 to 6 \( (1,1,1,1,1,1,0,0,...)\).
\end{enumerate}

\paragraph{The matrix \(D_n\).} The matrix $D_n\in \{0,1\}^{(4N+2)\times(10N+6)}$ is described by blocks in Figure \ref{figure:D_n}. 
The first column has a single 1 at the second entry, then there's a
square block of height and width $4N+1$, with $1$'s on its main diagonal and on the lower sub-diagonal. Afterwards come three identity matrices of dimensions $3N+1, 2N+1$ and $N+1$. The last column has a single 1 at the last entry (which could be thought of as an identity matrix of dimension $1$). All other entries of $D_n$ are $0$.

\begin{figure}[htp]
    \centering
    \includegraphics[scale=0.43]{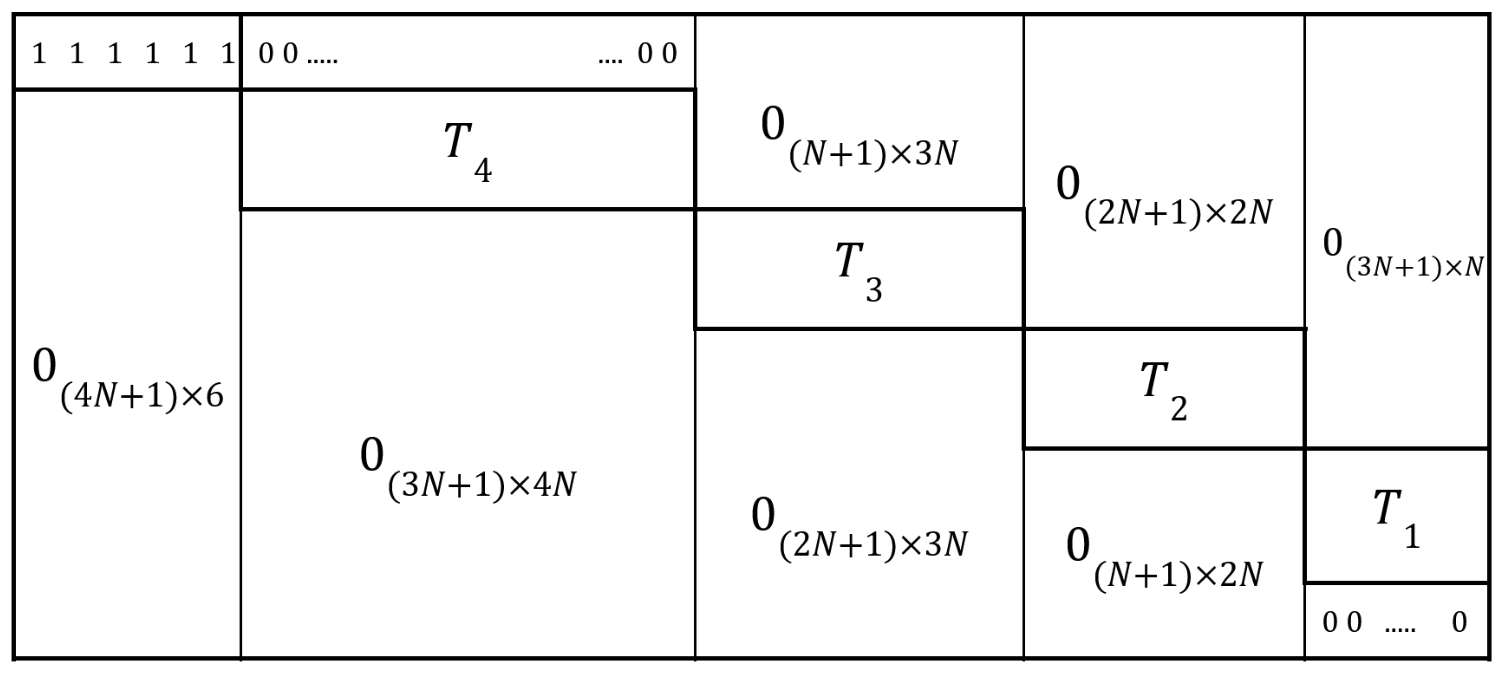}
    \caption{$S_n$. Recall that \(T_d = I_{N} \otimes \one_d\).} 
    \label{figure:S_n}
\end{figure}

\begin{figure}[htp]
    \centering
    \includegraphics[scale=0.45]{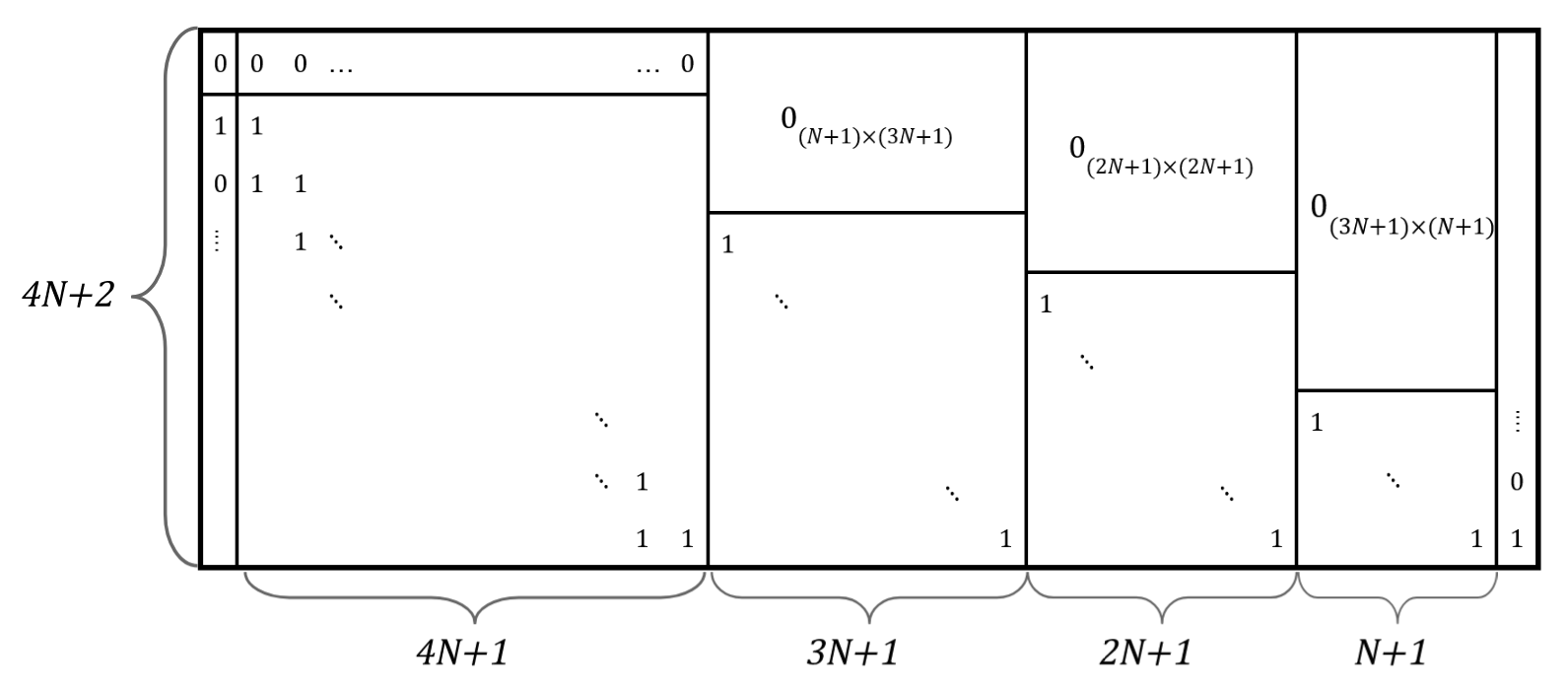}
    \caption{$D_n$. Recall that $A_n=S_n+D_n$.} 
    \label{figure:D_n}
\end{figure}

We show a formula for \(A_n\).

\begin{claim} \label{claim:a-n-formula}
\(A_n(i,j) = 1\) in the following cases:
    \begin{equation} \label{eq:a-n}
        \begin{cases}
             j\in[6] & i=1\\
             j\in\set{i-1,i} \cup \set{4i-1,...,4i+2} &i\in[2,N+1]\\
             j\in \set{i-1,i} \cup \set{i+3N+1} \cup \set{3i+N+1,...,3i+N+3}
             &i\in[N+2, 2N+1]\\
             j\in \set{i-1,i} \cup \set{i+3N+1} \cup \set{i+5N+2} \cup \set{2i+3N+3,2i+3N+4}
             &i\in[2N+2,3N+1]\\
             j\in \set{i-1,i} \cup \set{i+3N+1} \cup \set{i+5N+2} \cup \set{i+6N+3} \cup \set{i+6N+5}
             &i\in[3N+2,4N+1]\\
             j\in \set{i-1,i} \cup \set{i+3N+1} \cup\set{i+5N+2} \cup \set{i+6N+3} \cup\set{10N+6}
             &i=4N+2\\
        \end{cases}.
    \end{equation}
\end{claim}

\begin{proof}[Proof of Claim \ref{claim:a-n-formula}]
Observe that \(T_d(i,j) = 1\) if \(i = \lceil \frac{j}{d} \rceil\), i.e. \(j \in \set{(i-1)d+1,...,id}\). The \(T_d\) matrix in \(S_n\) begins with a column offset of \(6 + N\sum_{k=d+1}^4 k\),
and a row offset of \(1+(4-d)N\). 
For example, $T_4$ is shifted by six columns and one row, and therefore, in the rows corresponding to $T_4$, $S_n(i,j)=1$ iff $j-6\in \{(i-1-1)4+1,...,(i-1)4\}$. By calculating the shifts on $T_3, T_2$ and $T_1$ and rearranging we get the algebraic definition of $S_n$:\\
\begin{equation}
\label{eq:sn-formula}
S_n(i,j)=1\iff
\begin{cases}
    j\in[6] & i=1\\
    j\in \set{4i-1,...,4i+2}    & i\in[2,N+1]\\
    j\in \set{3i+N+1,...,3i+N+3}    & i\in[N+2, 2N+1]\\
    j =  \set{2i+3N+3,2i+3N+4}      & i\in[2N+2,3N+1]\\
    j = i+6N+3                      & i\in[3N+2,4N+1]
\end{cases}.
\end{equation}

The matrix $D_n$ consists of six diagonals (considering the bottom right entry as a length 1 diagonal). 
\begin{itemize}
    \item The first diagonal consists of the entries $(i,j)$ satisfying $j=i-1$.
    \item The second diagonal has entries $(i,j)$ where $j=i$ and $i\geq 2$.
    \item The third diagonal starts at the $(N+2, 4N+3)$-entry and includes all indices $(N+2+t, 4N+3+t)$ (for $t\in\{0,...3N\}$). In other words, it contains all $(i,j)$ s.t. $j=4N+3+t=i+3N+1$ (and $i\in\{N+2,...,4N+2\}$).
    \item The other diagonals are similarly calculated from Figure \ref{figure:D_n}.
\end{itemize}
Putting together these six diagonals with the formula for $S_n$ yields the full formula for $A_n$ (Equation \ref{eq:a-n}).
\end{proof}

We denote by \(j_t\) the leading entry index of the \(t\)-th row in \(S_n\). One may verify the following formula:
\begin{equation} \label{eq:jt}
    j_t = d_t t + \frac{(5-d_t)(4-d_t)}{2}N + 7-2d_t
\end{equation}
for \(d_t=5-\lceil\frac{t-1}{N}\rceil\) (i.e. \(d_t\) is the step "width" of row $t$, and row $t$ belongs to the rows corresponding to \(T_{d_t}\) in \(S_n\)).

\paragraph{The matrix \(B_n\).} The matrix \(B_n\in\{0,1\}^{\frac{3}{2}(N+1)\times(10N+6)}\) is depicted in Figure \ref{figure:B_n}. Note that the first identity matrix in \(B_n\) is placed on the second row, while the rest of the identity matrices begin at the first row.

It is easy to observe that \(B_n(i,j)=1\) in the following cases:
\begin{equation} \label{eq:b-n}
    \begin{cases}
      j\in \set {1, \frac{1}{2}(11N+5), 7N+4, \frac{1}{2}(17N+11) } & i=1\\
      j\in \set{4N+i,\frac{1}{2}(11N+3)+i, 7N+3+i, \frac{1}{2}(17N+9)+i } & i\in[2,\frac{3}{2}(N+1)]
    \end{cases}.
\end{equation}

\begin{figure}[htp]
    \centering
    \includegraphics[scale=0.4]{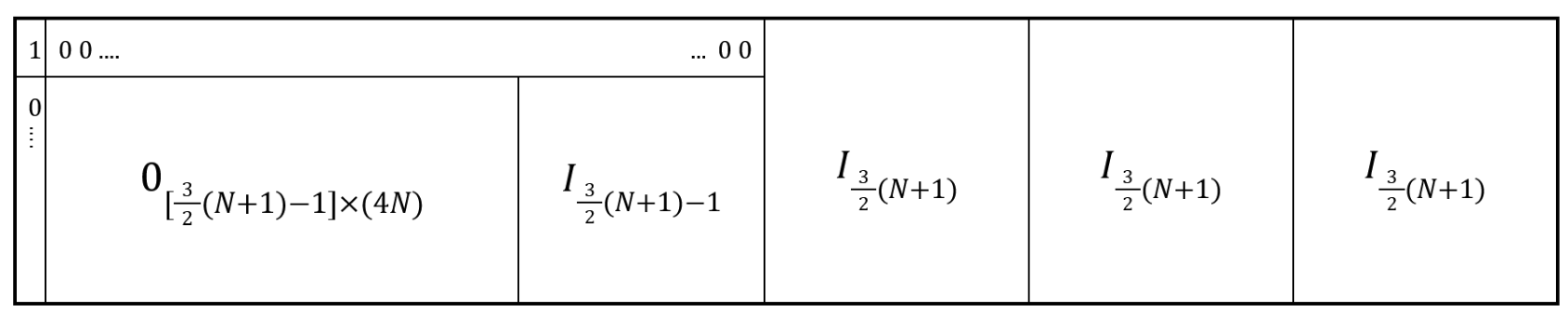}
    \caption{$B_n$}
    \label{figure:B_n}
\end{figure}

The bottom right block of $M_n$ was defined as $I_{\frac{N+1}{2}} \otimes \mathbf{1}_{3\times 2}$ but could actually be any $\frac{3}{2}(N+1)\times(N+1)$ matrix as long as it has row regularity 2 and column regularity 3.
$M_{18}$ is small enough to sketch (see Figure ~\ref{fig:M18}).\\

Before analyzing the counterexample we only need to confirm that all columns of $M_n$ have weight $3$, and that all rows have weight $6$ as claimed.

\begin{claim} \label{claim:regularity-of-m-n}
     Every row in $M_n$ has weight $6$, and every column weight $3$.
\end{claim}

\begin{proof}[Proof of Claim \ref{claim:regularity-of-m-n}]
    Clearly every row in $A_n$ has 6 ones (this can be verifies by looking at Equation \ref{eq:a-n}). Every row in $B_n$ has 4 ones (as indicated in Equation \ref{eq:b-n}), and the selection of the bottom right matrix is made to ensure that the weight of these rows is completed to 6.

    Let us verify that every \textbf{column} has \(3\) ones. The \(N+1\) rightmost columns (corresponding to the bottom left block) clearly have \(3\) ones, so we only need to show that the columns of 
$
\begin{bmatrix}
A_n\\
B_n
\end{bmatrix}
$
each have \(3\) ones. The diagonal $i$ starts at column 2 and continues to column $4N+2$, the diagonal $i+3N+1$ starts at columns $4N+3$ and ends a column before the diagonal $i+5N+2$ starts. Together with the diagonal $i+6N+3$ and entry $A_n(4N+2, 10N+6)$ the columns 2 to 10N+6 have another "layer" of $1$'s. The diagonal $i-1$ "covers" columns 1 to 4N+1, so in total, using Claim \ref{claim:a-n-formula}, columns $[2, 4N+1]$ have weight 3, while columns ${1}\cup[4N+2,10N+6]$ have weight 2. The weight of these columns is completed to 3 by the columns of $B_n$, which can be easily verified by its definition.
\end{proof}

We move on to prove \autoref{thm: matrix properties}.

\subsection{Proof of theorem 4.1} 

For the reader's convenience we restate the theorem.

\begin{T1}[Restated]
    \thmtext
\end{T1}

\begin{proof}~
\begin{enumerate}
\item For \autoref{subthm: valid choice} we use the following lemma (see  \hyperref[proof: valid-choices]{proof}).
    \begin{lemma} \label{lem:valid-choices}
        For \(t=1,2,...,4N+1\), choosing row \(t\) at iteration \(t\) is a valid choice for Algorithm \ref{alg:ess-finder}($M_n$).
    \end{lemma}

    Assuming \ref{alg:ess-finder} makes these choices, then after iteration $4N+1$ all the columns with 1's in those rows are zeroed out, yielding the matrix

    $$
\Tilde{M}_n=
\left[
\begin{array}{c|c}
0_{(4N+2)\times(10N+6)}&0_{(4N+2)\times(N+1)}\\
\hline\\
0_{\frac{3}{2}(N+1)\times(10N+6)}&I_{\frac{N+1}{2}} \otimes \mathbf{1}_{3\times 2}\\
\end{array}
\right]
$$
    
    So at iteration $4N+2$, row $4N+2$ (the last row of $A_n$) will certainly be chosen since it would have weight 0, while all rows $>4N+2$ (rows corresponding to $B_n$) will have weight 2. According to step \ref{step: no new vars} of Algorithm \ref{alg:ess-finder} this leads to the call $STRIP(A_n)$ which returns $A_n$ since all columns of $A_n$ have weight of at least 2. Hence, Algorithm \ref{alg:ess-finder} halts and outputs $A_n$.
    \item The removal of rows $1$ and $4N+3$ from $M_n$ (corresponding to the first row of $A_n=D_n+S_n$ and the first row of $B_n$) yields an ($m-2$)-row matrix that has a ``staircase form'', so  $Rank(M_n)\geq m-2$. This can be seen by looking at $D_n$  (Figure \ref{figure:D_n}) and at $B_n$ (Figure \ref{figure:B_n}). The first diagonal in $D_n$ continues into $B_n$'s first identity matrix. 
    We conclude that $dim(Ker(M_n))=n-Rank(M_n)\leq n-(m-2)=n/2+2$.
    \item $A_n$ has $10N+6$ columns and $4N+2$ rows, therefore $$dim(Ker(A_n))=10N+6-Rank(A_n)\geq6N+4\overset{(n=11N+7)}{=}n/2+\Omega(n) \gg n/2$$
\end{enumerate}
\end{proof}

\begin{proof}[Proof of Lemma \ref{lem:valid-choices}]\label{proof: valid-choices}
Let us introduce the following notation.
For any $i\in [m], j\in [n]$,
$$
wt(i,j):=\sum_{k=j}^{n}M_n(i,k) 
$$
This is the Hamming weight of \textbf{row} \(i\) restricted to entries \(j,\ldots ,n\).

Recall that \(j_t\) is the leading entry index of the \(t\)-th row in \(S_n\). For example, in Figure \ref{figure:j_t} we highlight the entry $S_n(3, j_3)$.

\begin{figure}[htp]
    \centering
    \includegraphics[scale=0.43]{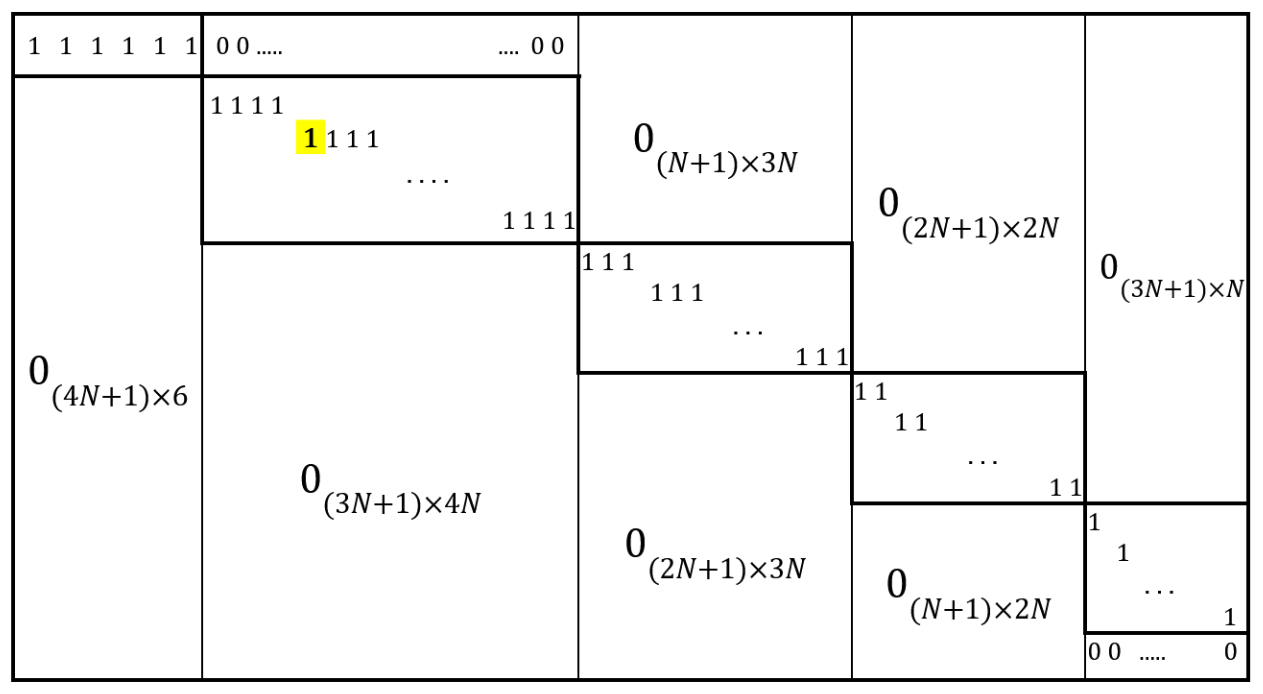}
    \caption{$S_n(3, j_3)$ highlighted}
    \label{figure:j_t}
\end{figure}

Assume by induction that rows $1,.., t-1$ were chosen in iterations $1,.., t-1$ of Algorithm \ref{alg:ess-finder}($M_n$) for some $t\in \{2,..., 4N+1\}$. Then columns $\set{1,...,j_t-1}$ were zeroed-out in those iterations since for every \(t\), the submatrix \(M_n(\set{1,...,t-1},\set{1,...,j_t-1})\) has at least one non-zero entry at every column. This can be seen by recalling that the top left block of $M_n$ is $A_n=D_n+S_n$ and by noting that every column in \(S_n(\set{1,...,t-1},\set{1,...,j_t-1})\) has at least one non-zero entry (for illustration look at Figure \ref{figure:j_t}). Moreover, the submatrix \(M_n(\set{1,...,t-1},\set{j_t,...,n})\) is the zero matrix, so no other columns were zeroed out in previous steps. To verify this claim, observe that by Claim \ref{claim:a-n-formula}, the last non-zero entry for every row $i\in \{1,...,4N+1\}$ comes from $S_n$, and that $S_n(\{1,...,t-1\},\{j_t,...,n\})$ is the zero matrix (see again Figure \ref{figure:j_t} for illustration).

To complete the proof we need to show that for all \(i>t\),
\begin{equation} \label{eq:lightest-row}
    wt(t,j_t)\leq wt(i, j_t),
\end{equation}
meaning that row $t$ is a valid choice for the algorithm at iteration $t\in [4N+1]$. 

If $t=4N+2$ then $wt(t,j_t)=0$ so \eqref{eq:lightest-row} clearly holds, so assume that $t<4N+2$. We prove \eqref{eq:lightest-row} by analyzing separately the rows $i \leq 4N+2$ (corresponding to $A_n$) and the rows $i > 4N+2$ (corresponding to $B_n$).

We begin with \(i \leq 4N+2\). Recall that $d_t$ is the number of ones in a row in a the block $T_d$ in $D_n$, such that $t$ is in that block. If $t$ is the last row of a block \(T_d\), i.e. \(t \in \set{N+1,2N+1,3N+1,4N+1}\), then \(wt(t,j_t)=d_t=wt(i, j_t)\) for every row $i\in (t, 4N+2]$. Let us demonstrate this by example. The row \(t=N+1\) has non-zeros only coming from $S_n$ (and not $D_n$). These are four non-zeros coming from $T_4$ so \(wt(t,j_t)=4\). Rows \(i\in [N+2,2N+1]\) have ones from both $S_n$ and $D_n$. In $S_n$, they get three $1$'s  from $T_3$. An additional \(1\) comes from the second identity component in \(D_n\). The remaining cases are similar and easy to verify.

For \(t \notin \set{N+1,2N+1,3N+1,4N+1}\), let \(t'\) denote the last row in the block of \(t\), namely, \(d_{t'}=d_t\) and for every \(i>t'\), \(d_i < d_t\). For all \(i\in(t, t']\), \(wt(i,j_t)\geq d_t=wt(t,j_t)\). For all $i\in(t', 4N+2]$, \(wt(i,j_t) \geq d_t\) since \(wt(i,j_t)\geq wt(i,j_{t'})=d_t\), thus proving the first item.

Thus far we showed that at iterations $t\in [4N+1]$, row $t$ was not heavier than any of the rows $i \in (t, 4N+2]$ (the rows below $t$ in $A_n$).
Now we prove that the first row in \(B_n\) is the lightest among all rows of \(B_n\), that is, $wt(4N+3,j_t)\leq wt(i, j_t)$ for any \(i \geq 4N+3\).
we denote by \(i'=i-4N-2\) the index of row \(i\) relative to \(B_n\). 
The value $wt(4N+3,j)$ decreases at columns $j\in\set{1, \frac{1}{2}(11N+3)+1, 7N+4, \frac{1}{2}(17N+9)+1}$ while the value $wt(i, j)$, decreases at $j\in\set{4N+i', \frac{1}{2}(11N+3)+i', 7N+3+i', \frac{1}{2}(17N+9)+i'}$ for all \(i>4N+3\). Thus concluding that \(wt(4N+3,j_t) \leq wt(i,j_t)\) for all \(i>4N+3\) and all \(t\).

At last we show that at iterations $t\in[4N+1]$, row $t$ of $A_n$ is no heavier than the first row of $B_n$
(i.e. \(wt(t,j_t)\leq wt(4N+3,j_t)\)). This is also a case analysis:
\begin{itemize}
    \item At row $t=1$,  $wt(1,1)=wt(4N+3,1)=6$.
    \item For $t\in[2, N+1]$, we prove that $wt(4N+3,j_t)\geq4=wt(t,j_t)$:
    
    The rows \(t\in[2, N+1]\) correspond to the block $T_4$ in $S_n$, therefore \(wt(t,j_t)=4\). By the definition of $S_n$ (Equation \ref{eq:sn-formula}), when \(t\in[2, N+1]\), then \(j_t=4t-1\), i.e. \(j_t\in[7, 4N+3]\). By the definition of $B_n$ (Equation \ref{eq:b-n}), the first row has weight of at least 2 for all \(j\leq 7N+4\). Adding weight 2 from the block \(I_{\frac{N+1}{2}}\otimes \mathbf{1}_{3\times 2}\) we get that \(wt(4N+3,j)\geq 4\) for all \(j \in [7, 4N+3]\).
    
    \item Similarly we can show that in rows $t\in[N+2, 2N+1]$, $wt(4N+3, j_t)\geq3=wt(t,j_t)$. These rows correspond to the block $T_3$ in $S_n$, therefore $wt(t,j_t)=3$ and $j_t\in [4N+7,7N+4]$, while for all $j\leq\frac{1}{2}(17N+11)$, $wt(4N+3,j)\geq3$.
    \item In rows $t\in[2N+2, 4N+1]$, $wt(4N+3, j_t)\geq2\geq wt(t, j_t)$. For all columns corresponding to $A_n$ ($j\in[10N+6]$), $wt(4N+3, j)\geq2$ since the first row of the block $I_{\frac{N+1}{2}}\otimes \mathbf{1}_{3\times 2}$ contributes 2 to the weight of row $4N+3$. Rows $t\in[2N+2, 4N+1]$ correspond to the blocks $T_2$ and $T_1$ in $S_n$ and therefore $wt(t, j_t)\in\{1,2\}$. 
\end{itemize}
\end{proof}

\section{The typical case}
We would like to emphasize that although our example may seem like a carefully constructed counterexample, it seems that a random low density matrix will fail as well. 
 In \cite{RU} the authors analyze the expected behavior of similar algorithms on random matrices, and conclude that none of them yield sub-quadratic encoding complexity. The nature of their analysis is heuristic and therefore cannot hold as a formal proof. Nevertheless, their results are backed with experimentation, so they are likely to have a holding in reality. Our experiments show that a random matrix will have a first component with too many message bits, with very high probability\footnote{The probability depends on the row and column weight distributions. The observation applies for example to (3,6)-regular matrices that are large enough (say $n>200$).}.

\section[Running Algorithm \ref{alg:decomposition} on M18]{Running Algorithm \ref{alg:decomposition} on \(M_{18}\)} \label{sec:running-the-alg}
As a warm up, let us run Algorithm \ref{alg:decomposition} on \(M_{18}\) depicted in Figure \ref{fig:M18}. There are many choices made by this algorithm, so we will show a sequence of choices that result in an output of components that do not describe the code \(Ker(M_{18})\).

\begin{enumerate}
    \item Let \(M=M_{18}\). The first time the Algorithm \ref{alg:decomposition} is called as a subroutine, it is called on \(M\) and returns \(A_{18}=M([6],[16])\), the first six rows in \(M\):
    \begin{enumerate}
        \item This is because for iterations \(i=1,2,\dots,6\) of Algorithm \ref{alg:decomposition}, the \(i\)-th row of \(M\) is a lightest row.
        \item After selecting these rows, in the sixth iteration \(V_c = \emptyset\) so after running step \(2(b)\), the STRIP procedure returns \(A_{18}\) which Algorithm \ref{alg:decomposition} outputs.
    \end{enumerate}
    This matrix has full rank so it is not a PESS.
    \item The second time Algorithm \ref{alg:decomposition} is called as a subroutine, it is called on \(M(\set{7,8,9},\set{17,18})\) (the grey part of Figure \ref{fig:M18}). One observes that \(M(\set{7,9},\set{17,18})\) is a valid output in this step.
    \item \(M(\set{7,9},\set{17,18})\) is a PESS, so we remove one row (say, the \(7\)-th row) and recursively run Algorithm \ref{alg:decomposition} on the matrix whose rows are those of \(A_{18}\) plus the row \(c_7 + c_9\) (the sum of the seventh and ninth rows of \(M_{18}\)).
    Assuming \ref{alg:decomposition} works properly, the components it returns when called on $A_{18}$ with the new row, will have at least 9 input bits, since there are 16 variables and 7 constraints. However, the residual matrix $M(\{8,9\},\{17,18\})$ has rank 1, so it also has an input bit. In total, the output components will have at least 10 input bits, while $M_{18}$ has only 9 (it has full rank). We conclude that the output components of $\ref{alg:decomposition}(M_{18})$ do not describe the code $Ker(M_{18})$.
\end{enumerate}

\bibliographystyle{plain}
\bibliography{bibliography.bib}
\end{document}